\documentclass[a4paper]{jpconf}
\usepackage{graphicx}
\usepackage{amsmath,amsfonts,amsthm,amssymb,cite}
\usepackage[pdfstartview=FitH]{hyperref}
\newcommand{\norm}[1]{\left\lVert#1\right\rVert}
\newcommand{\abs}[1]{\lvert#1\rvert}
\newtheorem{thm}{Theorem}[section]
\newtheorem{prop}[thm]{Proposition}

\def\C{{\mathbb{C}}}	

\def\R{{\mathbb{R}}}

\def\I{{\mathbb{I}}}

\def\mm{{\mathcal M}}	
\def\M{{\mathcal M}}		
\def\aa{{\mathcal A}}

\def\A{{\mathcal A}}

\def\hh{{\mathcal H}}

\def\pp{{\mathcal P}}

\def\X{{\mathcal X}}

\def\ss{{\mathcal S}}

\def\cinf{C^{\infty}\lp\mm\rp}

\def\lp{\left(}
\def\rp{\right)}

\def\xo0{\omega^0_x}
\def\yo0{\omega^0_y}

\def\xo0{x_\omega^0}
\def\yo0{y_\omega^0}

\def\pa{{\mathcal P}(\aa)}
\def\sa{{\mathcal S}(\aa)}

\def\X{{\mathcal X}}

\begin{document}
\title{Connes distance  and optimal transport}

\author{Pierre Martinetti}

\address{Dipartimento di matematica, universit\`a di Genova, via
  Dodecaneso 35, 16146 Genova}

\ead{martinetti@dima.unige.it}

\begin{abstract}
We give a brief overview on the relation between Connes spectral distance in
noncommutative geometry and the Wasserstein distance of order $1$ in
optimal transport. We first recall how these two distances
coincide on the space of probability measures on a Riemannian
manifold. Then we
work out a simple example on a discrete space,
showing that the spectral distance between arbitrary states does not coincide with the
Wasserstein distance with cost the spectral distance between pure states. \end{abstract}

\section{The metric aspect of noncommutative geometry}

Topology is the
minimal structure required for a set to be called \emph{space}, and
for its elements to become \emph{points}. Topology gives sense to the
notion of neighbourhood (and more generally to the notion
of open sets). The algebraic dual notion is that of
continuity (a continuous function being, by definition, such that the
inverse image of any open set is an open set). Actually this duality
goes quite far, since all the information of a topological space is contained within the
algebra of continuous functions defined on it. More precisely,
Gelfand's duality states that any complex \emph{commutative} $C^*$-algebra $\A$ is isomorphic to the algebra of
continuous functions vanishing at infinity on some locally compact
topological space - the space $\pa$  of pure states
of $\A$ - and conversely any locally compact topological space $\X$  is
homeomorphic to the space of pure states of the \emph{commutative} algebra $C_0(\X)$ of
continuous functions on $\X$ vanishing at infinity: 
\begin{equation}
 \A \simeq C_0(\pa), \quad \X \simeq \pp(C_0(\X)).\label{eq:6}
 \end{equation}

Recall that the pure states of an involutive algebra $\A$ are the extremal points of the set
of states, the latter being the  linear maps $\varphi$ on $\A$
wich are positive
- $\varphi(a^*a)\in\R^+$ -
and of norm $1$
(where $||\varphi||=\sup_{a\in\A}\frac{|\varphi(a)|}{||a||}$). In particular, a state of
the commutative algebra
$C_0(\X)$ is the integration with respect to a probability
measure $\mu$, with pure states given by Dirac $\delta$ measures (i.e.
evaluation at a point): 
$$\pp(C_0(\X))\ni\delta_x: x\to \int_\X f\delta_x = f(x) \,,\qquad S(C_0(\X))\ni\varphi:
f\to\int_\X f \,\text{d}\mu \quad \forall f\in C_0(\X).$$
 
 Connes' noncommutative geometry extends Gelfand's duality
 beyond topology, so that to encompass all the aspects of Riemannian
 geometry, in particular the metric. To do so, one needs more than an
 algebra: a \emph{spectral triple} \cite{Connes:1996fu} 
 consists in an involutive algebra $\A$ acting 
faithfully on an Hilbert space $\hh$, with $D$ a selfadjoint operator on $\hh$ such
that the commutator $[D, a]$ is bounded and  $a[D - \lambda\I]^{-1}$ is compact
for any $a\in \A$ and $\lambda\notin \text{ Sp } D$. 
  When 
a set of conditions (dimension, regularity, finitude, first order,
orientability) is satisfied, then one is able to characterize a
Riemannian manifold by purely spectral data\cite{connesreconstruct}:  
\begin{itemize}

\item {For $\M$ a compact Riemannian manifold, then
   \begin{equation}
(C^\infty\!(\M), \Omega^\bullet(\M), d+d^\dag)
\label{eq:9}
\end{equation}
is a spectral triple, where $\Omega^\bullet(\M)$ is the Hilbert space
of square integrable differential forms on $\M$ and $d+d^\dag$ is the Hodge-Dirac operator ($d$ the exterior derivative, $d^\dag$ its Hodge-adjoint).

\item  When $(\aa, \hh, D)$ is a spectral triple with $\A$ unital commutative, then there exists a compact Riemannian manifold $\M$ such that $\A=C^\infty(\M)$.}
\end{itemize}

By adding two extra-conditions (real structure and Poincar\'e
duality), the result is extended to spin manifolds.

Why is such a spectral characterization of manifolds interesting ?
Because the properties defining a spectral triple still make sense for a
noncommutative $\A$
\cite{Connes:1996fu}. A \emph{noncommutative
  geometry} is thus a spectral triple where the algebra $\A$ is
noncommutative. At the light of Gelfand duality, this is the geometrical object whose
algebra of functions defined on it is non commutative. As such it
cannot be a usual topological space (otherwise its algebra of
continuous functions would be commutative), but it rather appears as a
``space without points''.
\medskip

 \begin{eqnarray*}
\text{commutative spectral triple} &\rightarrow& \text{noncommutative
  spectral triple}\\
\updownarrow & & \downarrow \\
\text{Riemannian geometry} & & \text{non-commutative geometry}
\end{eqnarray*}
\medskip

However, always at the light of Gelfand duality, it is tempting to
consider the pure states of the
noncommutative algebra as the equivalent of points in the
noncommutative context. This is all the more appealing that the same formula that allows to retrieve the
Riemannian geodesic distance in Connes reconstruction theorem, 
also provides the space of states with a distance.

Explicitly, given a spectral triple $(\A,\hh, D)$, with $\A$ commutative or not, one defines on its state
space $\sa$ the \emph{spectral
  distance} \cite{Connes:1992bc}
\begin{equation}
d_D (\varphi, \psi) = \underset{a\in \A}{\sup}\{
\abs{\varphi(a) - \psi(a)}\, \slash \norm{[D,a]}\leq
1\} \qquad \forall \varphi, \psi\in \sa.\label{eq:8}
\end{equation}
 It is not difficult to check that $d_D$ has all the properties of a
 distance (zero if and only if $\varphi=\psi$, $d_D (\varphi, \psi)  =
 d_D (\psi, \varphi)$, triangle inequality), except that it may be
 infinite.  By a slight abuse of
 language, we still call it \emph{distance}. For an overview of
 explicit computation of this distance in various examples of
 commutative and noncommutative spectral triples, see \cite{Martinetti:2016aa}.

\section{Rieffel's remark and Wasserstein distance of order $1$}

Rieffel noticed in \cite{Rieffel:1999wq} that formula 
\eqref{eq:8}, applied to the spectral
triple of a Riemannian manifold \eqref{eq:9}, was nothing but the Wasserstein distance of order $1$ in
the theory of optimal transport, or more exactly a reformulation of
Kantorovich dual of the Wasserstein distance.

To see it, let us first remind what the Wasserstein (or Monge
Kantorovich distance) is. Let $\X$ be a locally compact Polish space, $c(x,y)$  a positive real function, the ``cost''. 
The minimal work $W$ required to transport the probability measure
$\mu_1$ to $\mu_2$ is
\begin{equation}
W(\mu_1,\mu_2)\doteq 
\inf _{\pi} \int_{\mathcal{X}\times\mathcal{X}} c(x, y)\;d\pi
\label{eqq:7}
\end{equation}
where the infimum is over all \emph{transportation plans},
i.e. measures $\pi$ on ${\mathcal{X}}\times {\mathcal{X}}$ with
marginals $\mu_1, \mu_2$. When the cost function $c$ is a distance $d$, then
\begin{equation}
W(\mu_1,\mu_2)\doteq 
\inf _{\pi} \int_{\mathcal{X}\times\mathcal{X}} d(x, y)\; d\pi
\label{eq:7}
\end{equation}
 is a distance (possibly infinite) on the space of probability measures on $\X$, called the
 Monge-Kantorovich (or Wasserstein) distance of
   order $1$.


In \cite{Kan42}, Kantorovich showed that Monge problem of minimizing
the cost (eq. \eqref{eqq:7}) had an equivalent dual formulation
(interpreted as maximizing a profit). Namely, $W(\mu_1, \mu_2)$ is
equal to 
\begin{equation}
\label{kanto1}
W(\varphi_1, \varphi_2) = \sup_{\norm{f}_{\mathrm{Lip}}\leq
  1}\left(\int_{\X}  f  d\mu_1 - \int_{\X} f  d\mu_2\right)
\end{equation}
where $\varphi_1, \varphi_2$ are the states of $C(\X)$ defined by the
measure $\mu_1, \mu_2$:
\begin{equation}
  \label{eq:10}
  \varphi_i(f)=\int_\X f d\mu_i \quad\forall f\in C(\X),\;  i=1,2;
\end{equation}
and the supremum in \eqref{kanto1} is on all the functions
$1$-Lipschitz with respect to the cost, that is
\begin{equation}
  \label{eq:11}
  f(x,y)\leq c(x,y) \quad \forall x,y \in\X.
\end{equation}
Let $\X=\M$ be a complete, connected, without boundary,  Riemannian
manifold. For any $\varphi, \tilde\varphi \in \ss(C_0(\M))$, 
$$W(\varphi, \tilde\varphi) = d_D(\varphi,\tilde\varphi) $$
where $W$ is the Wasserstein distance associated to the cost
$d_{\text{geo}}$, while $d_D$ is the spectral distance associated to 
$\left(C^\infty_0(\mm), \Omega^\bullet(\M), D= d+d^{\dag}\right)$.
That Kantorovich dual \eqref{kanto1} coincides with the spectral
distance \eqref{eq:8} then follows from the observation that the
supremum in the latter can be searched equivalently on selfadjoint
elements, for which one has  
\begin{equation}
\norm{[d+d^\dagger, f]^2} = \norm{f}^2_{\mathrm{Lip}}.
\label{eq:12}
\end{equation}
As pointed out in \cite{dAndrea:2009xr}, one has to be careful that
the manifold is complete, otherwise there is no guarantee that the
supremum on the $1$-Lipschitz functions in \eqref{kanto1} coincides with the supremum on $1$-Lipschitz functions
vanishing at infinity  in \eqref{eq:8}.

\section{Towards a theory of optimal transport in noncommutative
  geometry ?}

Connes spectral distance on a manifold coincides with Kantorovich dual
formulation of the Wasserstein distance of order $1$. It is quite
natural to wonder if the same is true is the noncommutative
setting. But there does not exist any 
``noncommutative Wasserstein distance'' of whom the spectral distance
would be the dual. Is it possible to build one ? More specifically,
given a spectral triple $(\A, \hh, D)$ with noncommutative $\A$, is
there a Wasserstein distance $W_D$ on $\sa$ such that its Kantorovich
dual is the spectral distance $d_D$ ? 
\begin{eqnarray*}
\text{{\bf Commutative case:}} & & \text{{\bf Noncommutative case:}}\\[5pt]
\text{Connes distance $d_D$} &\rightarrow&
\text{Connes distance $d_D$}\\
\uparrow & & \lvert \\
 \text{Kantorovich duality} & & \text{Kantorovich duality  ?} \\
\downarrow & & \downarrow\\
\text{Wassertein distance } W & &W_D \text{ for some noncommutative cost ?}\\
\text{ with } d_D(\delta_x, \delta_y) \text{ as a cost function}& &  
\end{eqnarray*}

In the commutative case $\A=\cinf$, one retrieves the cost function as the
Wasserstein distance between pure states:
\begin{equation*}
  \label{eq:4}
  W(\delta_x, \delta_y) = c(x,y).
\end{equation*}
So it is tempting to define $W_D$ on the whole space of states $\sa$
as the Wasserstein distance associated with the cost $c$ defined by
the spectral distance on the space of pure states $\pa$, that is 
\begin{equation}
c(\omega_1,\omega_2) :=d_D(\omega_1, \omega_2)\label{eq:14}
\end{equation}

We worked out this construction in \cite{Martinetti:2012fk},
restricting to unital separable $C^*$-algebras, for which it is known
that a state is a probability measure on $\pa$ \cite[p.144]{Bratteli:1987fk}
{\footnote{This may be true in general, but for safety we restrict to
    this well known case.}}, namely to any 
$\varphi\in\sa$, there exists a (non-necessarily unique)  probability
measure $\mu\in\text{Prob}(\pa)$ such that 
\begin{equation}
{\varphi(a) = \int_{\pa} \hat{a}(\omega)\, d\mu(\omega)}\label{eq:15}
\end{equation}
  where $\hat{a}(\omega) \doteq \omega(a)$ denotes the evaluation at
  $\omega\in\pa$ 
   of an element $a$ of $\A$, viewed as a function on $\pa$.
The Wasserstein distance on $\sa$ associated
with the cost \eqref{eq:14}
has Kantorovich-dual formulation
\begin{equation}
  \label{eq:2}
  \ {W_D(\varphi, \tilde\varphi) \doteq  \sup_{a\in \text{Lip}_D(\A)} \left\{\abs{\int_{\pa} \hat{a}(\omega) \,d\mu(\omega)
    - \int_{\pa} \hat a(\omega) \,d\tilde\mu(\omega)}\right\}},
\end{equation}
where
\begin{equation}
  \label{eq:3bis}
 {\text{Lip}_D(\A)\doteq \{a\in \A \text{ such that }
\abs{\omega_1(a) - \omega_2(a)} \leq d_D(\omega_1, \omega_2) \;
  \forall \omega_1, \omega_2\in\pa\}} 
\end{equation}
is the set of element of $\A$ that are $1$-Lipschitz with respect to
the cost \eqref{eq:14}.

It is not difficult to show that the Wasserstein distance provides
a upper bound to the spectral distance \cite[Prop. III.1]{Martinetti:2012fk},
\begin{equation}
d_D(\varphi, \tilde\varphi) \leq W_D(\varphi,\tilde\varphi) \quad
\forall \varphi, \tilde\varphi\in\sa.
\label{eq:16}
\end{equation}
The equality holds on any subset of $\sa$ given by a convex
linear combination of two pure states: fixed $\omega_1,
\omega_2\in\pa$, one denotes $\varphi_\lambda:= \lambda\omega_1 +
(1-\lambda)\omega_2$. Then 
\begin{equation}
  \label{eq:18}
  d_D(\varphi_{\lambda_1}, \varphi_{\lambda_2})=
  W_D(\varphi_{\lambda_1}, \varphi_{\lambda_2}) \quad \forall
  \lambda_1, \lambda_2\in\R.
\end{equation}
This shows in particular that $W_D=d_D$ on the whole of $\sa$ if $\A=M_2(\C)$, since the pure state
space of the algebra of $2\times 2$ matrices is homeomorphic to the
$2$-sphere, so that the space of states is the $2$-ball, and any two
states $\varphi_1, \varphi_2$ can always be decomposed as two convex
linear combinations $\varphi_{\lambda_1}, \varphi_{\lambda_2}$ of the same two pure states. 
\medskip

However the two distances are not equal in
general, as can be seen in the following counter-example, taken from
\cite[\S 7]{Rieffel:1999ec}. Consider $\A=\C^3$ acting on $\hh=\C^3$ as
a diagonal matrix,
\begin{equation}
  \label{eq:19}
  \pi(z_1,z_2,z_3) :=\left(
    \begin{array}{ccc}
      z_1&0 &0\\0&z_2&0\\0&0&z_3\end{array}\right)
\quad \forall\, (z_1, z_2, z_3)\in\C^3,
\end{equation}
and take as Dirac operator 
\begin{equation}
  \label{eq:20}
  D=\left(
    \begin{array}{ccc}
      0&0&\alpha\\ 0&0&\beta\\ \alpha&\beta&0
    \end{array}\right)\quad \alpha, \beta\in \R^+.
\end{equation}
There are three pure states $\delta_i$ for $\A$, defined as
\begin{equation}
  \label{eq:22}
  \delta_i(z_1, z_2, z_3) = z_i \quad i=1,2,3.
\end{equation}
So the space of states is the plain triangle with summit $\delta_1, \delta_2,
\delta_3$.

 By \eqref{eq:18}, one has that $W_D$ coincides with $d_D$
on each edge of the triangle. But the two distances do not agree on
the whole triangle.

\begin{prop}
\label{prop:WD}
  Let $\varphi$, $\varphi'$ be states  in
${\cal S}(\C^3)$,
\begin{equation}
  \label{eq:23}
  \varphi= \lambda_1 \delta_1 + \lambda_2\delta_2 +
  (1-\lambda_1-\lambda_2)\delta_3,\qquad
\varphi'= \lambda'_1 \delta_1 + \lambda'_2\delta_2 +
  (1-\lambda'_1-\lambda'_2)\delta_3
\end{equation}
where $\lambda_i, \lambda'_i\in\R$, $i=1,2$ are such that $\Lambda_1:=\lambda_1 -\lambda'_1$ and
$\Lambda_2:=\lambda_2 -\lambda'_2$ have the same sign. Then 
\begin{equation}
  \label{eq:25}
  W_D(\varphi, \varphi')= \frac{|\Lambda_1|}{\alpha} +  \frac{|\Lambda_2 |}{\beta} 
\end{equation}
while
\begin{equation}
  \label{eq:26}
  d_D(\varphi, \varphi')=  \sqrt{\frac{\Lambda_1^2}{\alpha^2} + \frac{\Lambda^2}{\beta^2}}
\end{equation}
\end{prop}
\begin{proof}
   The cost function \eqref{eq:14} is obtained computing explicitly the
   spectral distance between pure states  (see e.g.
 \cite[Prop. 7]{Iochum:2001fv}):
  \begin{equation}
    \label{eq:17}
    d_D(\delta_1, \delta_2)= \sqrt{\frac 1{\alpha^2}+ \frac
      1{\beta^2}},\qquad
    d_D(\delta_1, \delta_3)= \frac 1{\alpha},\qquad
    d_D(\delta_2, \delta_3)= \frac 1{\beta}.
  \end{equation}
The Lipschitz ball \eqref{eq:3bis} is thus
\begin{equation}
  \label{eq:27}
  \text{Lip}_D(\C^3)=\left\{(z_1, z_2, z_3)\in\C^3\,;\, |z_1-z_2|\leq  \sqrt{\frac 1{\alpha^2}+ \frac
      1{\beta^2}}, \,|z_1-z_3|\leq \frac 1{\beta},\, |z_2-z_3|\leq \frac 1{\alpha}\right\}.
\end{equation}
  For any $a=(z_1, z_2, z_3)\in \C^3$, one has \begin{equation}
    \label{eq:24}
    \varphi(a) -\varphi'(a) = \Lambda_1 (z_1 - z_3) +
    \Lambda_2 (z_2-z_3).
  \end{equation}
Therefore
  \begin{equation}
    \label{eq:28}
    W_D(\varphi, \varphi')\leq \frac{|\Lambda_1|}{\alpha} + \frac{|\Lambda_2|}{\beta}.
  \end{equation}
Since $\Lambda_1 $ e $\Lambda_2 $ have the
same sign, this upper bound is attained by $a_0$ in $\text{Lip}_D(\C^3)$ defined by
\begin{equation}
  \label{eq:29}
  z_1=\frac 1{\alpha},\quad z_2=\frac 1{\beta},\quad z_3=0.
\end{equation}

To prove \eqref{eq:26}, one computes the commutator of $a=(z_1, z_2,
z_3)$ with $D$ (see e.g. \cite{Iochum:2001fv})
\begin{equation}
  \label{eq:31}
  ||[D,a]||= \sqrt{\alpha^2|z_3-z_1|^2 + \beta^2|z_3-z_2|^2}.
\end{equation}
By subtracting $z_1\I$, one can always assume that $z_3=0$. The
commutator condition $||[D,a]||\leq 1$ becomes $\alpha^2|z_1|^2 +
\beta^2|z_2|^2\leq 1$, which is equivalent to
\begin{equation}
  \label{eq:32}
  |z_2|\leq \sqrt{\frac{1-\alpha^2|z_1|^2}{\beta^2}}.
\end{equation}
  For such $a$, one obtains from \eqref{eq:24}
  \begin{align}
    \label{eq:33}
    |\varphi(a) - \varphi'(a)|&\leq \Lambda_1 |z_1| + \Lambda_2|z_2|\\
& \leq \Lambda_1 |z_1| + \Lambda_2\sqrt{\frac{1-\alpha^2|z_1|^2}{\beta^2}}.
  \end{align}
On $[1, \frac 1{\alpha}]$, the function $f(x)= \Lambda_1 x
+\Lambda_2\sqrt{\frac{1-\alpha^2x^2}{\beta^2}}$ reaches its maximum
when $f'$ vanishes, that is for
\begin{equation}
 x_0=\frac{\Lambda_1}{\alpha^2\sqrt{\frac{\Lambda_1^2}{\alpha^2} + \frac{\Lambda_2^2}{\beta^2}}}.
\label{eq:34}
\end{equation}
This maximum,
\begin{equation}
  \label{eq:35}
  f(x_0)= \sqrt{\frac{\Lambda_1^2}{\alpha^2} + \frac{\Lambda^2}{\beta^2}},
\end{equation}
 is an upper bound for the spectral distance, reached by the
element $a=(x_0, f(x_0), 0)$.
\end{proof}

 \section{Conclusion and outlook}

On a manifold, Connes spectral distance between arbitrary states coincides
with the Wasserstein distance of order $1$ with cost the geodesic
distance. On an arbitrary spectral triple, the spectral distance $d_D$
on the space of states $\sa$  (viewed as the convex hull of the pure
states space $\pa$) 
does not coincide with the Wasserstein distance $W_D$ with cost 
function $d_D$ (on $\pa$). However, $W_D$ always provides an upper
bound to the spectral distance, and the two distances coincides on any
convex combination of two fixed pure states. 

This shows that the interpretation of the spectral distance
as a Wasserstein distance is more involved that could be initially
thought, although the intriguing example worked out in this paper ($W_D$ is the sum of the opposite and adjacent
sides of a right a triangle, $d_D$ is the length of the hypothenuse) suggests
that there might exist a simple relation between the two distances. 

Let us also mention that there do exists a formulation of the spectral
distance as an infimum rather than a supremum (an analogue to ``dual
of the dual'' formula of the Wasserstein distance in optimal transport),
whose possible interpretation as a noncommutative cost still has to be elucidated\cite{dAndrea:2009xr}.





\section*{References}
\begin{thebibliography}{9} 

\bibitem{Bratteli:1987fk}
O.~Bratteli and D.~W. Robinson, \emph{Operator algebras and quantum statistical
  mechanics 1}, Springer, 1987.

\bibitem{Connes:1996fu}  A.~Connes, \emph{Gravity coupled with matter and the foundations of
  noncommutative geometry}, Commun. Math. Phys. \textbf{182} (1996), 155--176.

\bibitem{Connes:1992bc}
A.~Connes and J.~Lott, \emph{The metric aspect of noncommutative geometry},
  Nato ASI series B Physics \textbf{295} (1992), 53--93.

\bibitem{Connes:1994kx}
Alain Connes, \emph{Noncommutative geometry}, Academic Press, 1994.

\bibitem{connesreconstruct}
Alain Connes, \emph{On the spectral characterization of manifolds}, J. Noncom. Geom.
  \textbf{7} (2013), no.~1, 1--82.

\bibitem{dAndrea:2009xr}
Francesco D'Andrea and Pierre Martinetti, \emph{A view on optimal transport
  from noncommutative geometry}, SIGMA \textbf{6} (2010), no.~057, 24
pages.

\bibitem{dAndrea:2009xr}
Francesco D'Andrea and Pierre Martinetti, in preparation.

\bibitem{Iochum:2001fv}
Bruno Iochum, Thomas Krajewski, and Pierre Martinetti, \emph{Distances in
  finite spaces from noncommutative geometry}, J. Geom. Phy. \textbf{31}
  (2001), 100--125.

\bibitem{Kan42}
L.~V. Kantorovich, \emph{On the transfer of masses}, Dokl. Akad. Nauk. SSSR
  \textbf{37} (1942), 227--229.

\bibitem{Martinetti:2012fk}
P.~Martinetti, \emph{Towards a {M}onge-{K}antorovich distance in noncommutative
  geometry}, Zap. Nauch. Semin. POMI \textbf{411} (2013).

\bibitem{Martinetti:2016aa}
P.~Martinetti, \emph{From {M}onge to {H}iggs: a survey of distance computation
  in noncommutative geometry}, Contemporary Mathematics \textbf{676} (2016),
  1--46.

\cite{Martinetti:2012fk}\bibitem{Rieffel:1999wq}
Marc~A. Rieffel, \emph{Metrics on states from actions of compact groups},
  Documenta Math. \textbf{3} (1998), 215--229.

\bibitem{Rieffel:1999ec}
Marc~A. Rieffel, \emph{Metric on state spaces}, Documenta Math. \textbf{4}
  (1999), 559--600.




\end{thebibliography}
 \end{document}